\documentclass[prx,aps,twocolumn,notitlepage,superscriptaddress,showpacs,nofootinbib]{revtex4-2}

\usepackage{enumerate,appendix}
\usepackage{qcircuit}
\usepackage{amsmath, amsthm, amssymb,commath}
\usepackage{color,calc,graphicx}
\usepackage[usenames,dvipsnames,svgnames,table,cmyk,hyperref]{xcolor}
\usepackage[colorlinks]{hyperref}
\usepackage{optidef}
\hypersetup{
	colorlinks = true,
	urlcolor = {blue},
	citecolor = {blue},
	linkcolor= {blue}
}

\newcommand{\GKB}[1]{{#1}}
\newcommand{\YJ}[1]{{#1}}

\usepackage{graphicx}
\usepackage{amsmath}
\usepackage{latexsym}
\usepackage{bbm}

\usepackage[charter,cal=cmcal,sfscaled=false]{mathdesign}
\usepackage{booktabs}
\usepackage{multirow} 
\usepackage{dcolumn}
\usepackage{mathrsfs}
\usepackage{csvsimple-l3}

\usepackage{threeparttable}

\def \be {\begin{equation}}
\def \ee {\end{equation}}

\newcommand{\ket}[1]{|#1\rangle}
\newcommand{\bra}[1]{\langle#1|}

\def\>{\rangle}
\def\<{\langle}

\newtheorem{theorem}{Theorem}
\newtheorem{lemma}[theorem]{Lemma}

 \newcommand{\yo}[1]{{#1}}

\begin{document}

\title{Measurement-free code-switching for low overhead  quantum computation using permutation invariant codes}

\author{Yingkai Ouyang}
\email{y.ouyang@sheffield.ac.uk}
\affiliation{School of Mathematical and Physical Sciences, University of Sheffield, Sheffield, S3 7RH, United Kingdom}

\author{Yumang Jing}
\affiliation{Center for Engineered Quantum Systems, Dept. of Physics \& Astronomy, Macquarie University, 2109 NSW, Australia}

\author{Gavin K. Brennen}
\affiliation{Center for Engineered Quantum Systems, Dept. of Physics \& Astronomy, Macquarie University, 2109 NSW, Australia}

\begin{abstract}
Transversal gates on quantum error correction codes have been a promising approach for fault-tolerant quantum computing, but are limited by the Eastin-Knill no-go theorem. Existing solutions like gate teleportation and magic state distillation are resource-intensive. We present a measurement-free code-switching protocol for  universal quantum computation, switching between a stabiliser code for transversal Cliffords and a permutation-invariant (PI) code for transversal non-Cliffords
that are logical $Z$ rotations for any rational multiple of $\pi$. The novel non-Clifford gates enabled by this code-switching protocol provide for a lower gate count implementation of a universal gate set relative to the Clifford$+T$ gate set. To achieve this, we present a protocol for performing controlled-NOTs between the codes using near-term quantum control operations that employ a catalytic bosonic mode. \GKB{We also present a new class of PI codes with tunable code distance, supporting transversal non-Clifford gates, and demonstrate their reduced gate count overhead relative to a comparable stabilizer code to stabilizer code switching scheme.  }
\end{abstract}

\maketitle

\section{Introduction}

Transversal gates \cite{EastinKnill-2009-PRL} on quantum error correction (QEC) codes are the cornerstone of many promising approaches to realising a fault-tolerant quantum computer \cite{campbell2017roads}. 
Transversal gates have an elegant structure which makes them fault-tolerant since they act on groups of physical qubits in parallel, thus avoiding a catastrophic spread of single errors within a code block.
While using transversal gates for fault-tolerant quantum computation is attractive, 
for any given QEC code,
the Eastin-Knill no-go theorem \cite{EastinKnill-2009-PRL} forbids
a universal set of transversal gates.
 
In light of the Eastin-Knill no-go theorem, the path towards realising fault-tolerant logical non-Clifford gates 
cannot rely on a single QEC code that implements universal quantum gates transversally. 
Instead, we need alternative methods.
One such approach is gate teleportation \cite{GCh99,ZLC00}, which consumes magic states. 
To ensure fault tolerance, the magic states must have sufficiently high quality, achieved by distilling cleaner magic states from a larger quantity of noisy ones \cite{bravyi2005universal,bravyi2012magic,krishna2019towards,campbell2017unified,campbell2017unifying,eastin2013distilling,hastings2018distillation}. 
However, the magic state distillation procedure can be costly both in terms of computation time and the number of qubits used \cite{campbell2017-magic-state-factory,cost-of-universality-PRXQuantum.2.020341,Litinski2019magicstate}.

Code-switching \cite{poulin-code-switch-PhysRevLett.113.080501,Bombín_2016,kubica-PhysRevA.91.032330,poulsen2017fault} presents a conceptually simple alternative to gate-teleportation using distilled magic states for enabling universal fault-tolerant quantum computations. 
In code-switching, different quantum error correction codes are employed for logical Clifford and non-Clifford gates, which allows all these logical gates to be performed transversally.
While code-switching offers a straightforward approach to achieve universal quantum gates fault-tolerantly, it can incur significant measurement overhead \cite{cost-of-universality-PRXQuantum.2.020341}.

For instance, in code-switching between two stabiliser codes \cite{butt2024fault,pogorelov2024experimental,huang2023graphical,daguerre2024code}, 
one performs multiple fault-tolerant stabiliser measurements to transform one stabiliser code into another, and the number of these measurements increases with the number of errors that the stabiliser codes can correct. 
Ref.~\cite{heußen2024efficientfaulttolerantcodeswitching} 
uses two single transversal measurements instead of multiple fault-tolerant stabiliser measurements. 
Another solution, amenable to surface codes, achieves a fault tolerant $CCZ$ gate using local transversal gates and code deformations over a time that scales with the size of the qubit array \cite{Brown_2019}.

A recent approach to measurement-free code-switching uses multi-qubit non-Clifford gates that are not transversal \cite{heussen2024measurement}. While elegant, compiling such non-transversal non-Clifford gates via near-term high-fidelity quantum control operations can be a challenge.
Hence, there remains the question regarding the possibility of measurement-free code-switching using only transversal gates.

In our paper, we show how to perform high fidelity quantum computing with a universal gate set using 
a measurement-free code-switching protocol.
In our protocol, we switch between QEC codes that encode a single logical qubit: a stabiliser code performs transversal logical Clifford gates and a permutation-invariant (PI) code that performs 
transversal logical non-Clifford gates.
The types of such transversal non-Clifford gates include logical $Z$ rotations for any rational multiple of $\pi$. 
The novel transversal non-Clifford gates that we introduce allow the implementation of gates in the binary icosahedral group and also an approximate $\tau_{60}$ gate \cite{kubischta2023family}   to allow the implementation of the most efficient known single-qubit universal gate set \cite{PARZANCHEVSKI2018869}.

At the core of our proposed code-switching algorithm 
is 
(1) a measurement-free state-teleportation protocol 
and that uses only logical CNOT gates, or equivalent gate set, between a stabiliser code and a PI code and 
(2)
our compilation of 
controlled-NOT gates that act between a stabiliser code and a PI code in terms of near-term quantum control.
The requirements of our protocol are modest.
First, we need strong, linear, coupling between the spins and the mode. In the case of coupling to a quantized cavity mode such as an optical or microwave cavity mode, this translates into the requirement for high co-operativity.
Second, we require that the stabiliser code admits transversal implementations of logical Clifford gates. Third, the stabiliser code should be an {\em even-odd quantum code} meaning its logical zero (one) codeword can be written as a superposition of even(odd)-weight computational basis states. That is, we can write the logical zero as 
$|0_L\> = \sum_{ {\bf x} \in \{0,1\}^n , |{\bf x}|{\ {\rm even}}} a_{\bf x} |{\bf x}\>$ and the logical one as
$|1_L\> = \sum_{ {\bf x} \in \{0,1\}^n , |{\bf x}|{\ {\rm odd}}} b_{\bf x} |{\bf x}\>$
for some complex coefficients $a_{\bf x}, b_{\bf x}$.
As explained below, stabiliser codes with the second two properties include 2D-color codes with even stabiliser weights and an odd number of qubits \cite{Bombin_2015,campbell2017roads} such as the Steane code \cite{Ste96},
and also Bacon-Shor codes \cite{shorcode-PhysRevA.52.R2493,bacon-PhysRevA.73.012340,baconshor} that are concatenations of odd length repetition codes. Most of the PI codes we consider here will also satisfy the even-odd  code property, though we show that it is possible to switch from an even-odd stabilizer code to a non even-odd PI code using a non-linear, dispersive coupling to a mode. 

\section{Code switching via swapping}

\subsection{Transversal gates on PI codes}
\label{introPIcodes}
PI codes \cite{Rus00,PoR04,ouyang2014permutation,ouyang2015permutation,OUYANG201743,movassagh2020constructing} allow QEC to be done on quantum states that are invariant under any permutation of the underlying particles. 
PI codes are attractive for various reasons. 
First their controllability by global fields could allow for their scalable physical implementations \cite{johnsson2020geometric} in near-term devices such as  trapped ions or ultracold atoms where addressability is challenging due to  cross-talk. 
Second, PI codes can correct deletions, i.e. erasures at unknown locations \cite{ouyang2021permutation,ouyang2023quantum}, along with insertion errors \cite{hagiwara2021four,shibayama2021equivalence,shibayama2022equivalence}, which conventional QEC codes cannot correct.

The potential to implement PI codes for applications such as quantum sensing \cite{ouyang2019robust,ouyang2023quantum} and quantum storage \cite{ouyang2019quantum} has been explored. 
However, it was only recently recognised that PI codes can also enable the transversal implementation of logical non-Clifford gates \cite{kubischta2023not,aydin2023family}. 
For simplicity, we first focus on two PI codes both of which encode a single logical qubit and have distance three, one on 7 qubits \cite{kubischta2023not,aydin2023family}, and another on 11 qubits \cite{kubischta2023not}. Both of these codes lie within the family of PI codes recently introduced by Aydin {\it et al.} \cite{aydin2023family}. For a definition see  Appendix~\ref{app:AAB}.

Permutation invariant quantum states on $N$ qubits are superpositions of the Dicke states 
$| D^N_w\>  ={\binom N w}^{-1/2}
\sum_{\substack{ x_1, \dots, x_N \in \{0,1\} \\ x_1+ \dots + x_N = w}}
|x_1, \dots, x_N\>,$
of weights $w = 0,\dots, N$. The Dicke states span the $N+1$ dimensional maximum spin ($J=N/2$) angular momentum space of the qubits.  
The 7-qubit PI code that we consider has support on four Dicke states \cite[Example 4]{aydin2023family} 
and is a multi-spin analog of the Gross code on a single spin~\cite[Eq. (14)]{gross-code}.
This 7-qubit code has logical codewords
\begin{align}
|0_{\rm pi7}\> &\coloneqq 
\sqrt{3/10}    |D_0^7\>
+
\sqrt{7/10}    |D_5^7\>\notag\\
|1_{\rm pi7}\> &\coloneqq 
\sqrt{7/10}    |D_2^7\>
-
\sqrt{3/10}    |D_7^7\>,\label{pi7}
\end{align}
and has distance three, but note that since the logical code words are superpositions of Dicke states with different parity weights, this is not an even-odd code.
 An 11-qubit PI code with distance three that was introduced in Ref~\cite{kubischta2023not} has logical codewords given by
\begin{align}
    |0_{\rm pi11}\>&= 
    \frac{1}{4} (
    \sqrt 5    |D^{11}_0\>
     + \sqrt {11} |D^{11}_8\>)
    \\
    |1_{\rm pi11}\>&=
    \frac{1}{4} (
    \sqrt {11}    |D^{11}_3\>
     + \sqrt 5 |D^{11}_{11}\>).
\end{align}
This Kubischta-Teixeira PI code supports a transversal implementation of the logical $T=Z(\pi/4)$ gate,
where 
$    Z(\theta) \coloneqq
    |0\>\<0| + |1\>\<1| e^{i\theta}$
denotes a qubit rotation operator about the $Z$-axis.

We denote a logical $ Z(\theta)$ gate as
$\bar Z_\Gamma(\theta)$ to denote a code logical operator on a set of physical qubits defined by a subsystem $\Gamma$ that has the following action on the codespace spanned by  logical codewords $|0_\Gamma\>$ and $|1_\Gamma\>$:
\begin{align}
    \bar Z_\Gamma(\theta) 
    (c_0 |0_\Gamma\> + c_1 |1_\Gamma\>)
    = 
     c_0 |0_\Gamma\> + c_1 e^{i\theta }|1_\Gamma\> .
\end{align}
Here, we mostly use $A$ to denote a subsystem that uses a stabiliser code, and $B$ to denote a subsystem that uses a PI code, but in general, we only require the codes in subsystems $A$ and $B$ to be even-odd codes.
The 7-qubit PI code admits a logical non-Clifford gate 
$\bar Z_B(4\pi/5) = Z(2\pi/5)^{\otimes 7} $ that is transversal.
 
Now let us denote the transversal operators on subsystem $\Gamma$ as
$\bar X_\Gamma \coloneqq X^{\otimes |\Gamma|}$,
$\bar Y_\Gamma \coloneqq Y^{\otimes |\Gamma|}$,
and
$\bar Z_\Gamma \coloneqq Z^{\otimes |\Gamma|}$, where $X,Y,Z$ denote qubit Pauli matrices
and $|\Gamma|$ is the number of qubits in $\Gamma$.  
Since
$\bar Z_B \bar X_B |0_{\rm pi7}\>
=|1_{\rm pi7}\>$
and
$\bar Z_B \bar X_B |1_{\rm pi7}\>
=- |0_{\rm pi7}\>$, the logical $-iY$ operator on the 7-qubit PI code is $X_{\rm pi7} \coloneqq \bar Z_B \bar X_B 
= -i \bar Y_B$.

The Pollatsek-Ruskai code is another 7-qubit PI code of distance three which is distinct from the 7-qubit code in \eqref{pi7}, since it is an even-odd code \cite{PoR04}.
The Pollatsek-Ruskai 7-qubit PI code admits transversal gates in the group $\mathrm{2I}$, the binary icosahedral group with order $|\mathrm{2I}|=120$ \cite{KT-2I}. 
For this code, the transversal $X$ and transversal $Z$ operators are also the logical $X$ and $Z$ operators, and the transversal $F$ gate is also the logical $F^*$ gate, where $F=H Z(-\pi/2)$ and $H$ denotes the Hadamard operator. 
This code can also implement the logical $\Phi^\star$ code with transversal $\Phi$ gates, where 
$\Phi$ and $\Phi^\star$ are non-Clifford gates given by 
\begin{align}
\Phi=\frac{1}{2}
\begin{pmatrix}
    \zeta    &     1\\
    -1 & \zeta^* \\
\end{pmatrix},\quad
\Phi^\star=\frac{1}{2}
\begin{pmatrix}
    \zeta^\star     &     1\\
    -1 &(\zeta^\star)^*
\end{pmatrix},
\end{align}
where 
$\zeta =  \varphi + i \varphi^{-1}$,
$\zeta^\star =  \varphi^\star + i (\varphi^\star)^{-1}$,
$\varphi = (1+\sqrt 5)/2$
and $\varphi^\star = (1-\sqrt 5)/2$.
Together, the logical $X$,
logical $Z$, logical $F$, and logical $\Phi^\star$ gates 
generate the algebra of a binary icosahedral group $\mathrm{2I}$.
Combined with the $\tau_{60}$ gate given by
\begin{align}
    \tau_{60} \coloneqq 
\frac{1}{\sqrt{5\varphi +7}}\begin{pmatrix}
    2+\varphi & 1 - i \\
    1+i & -2-\varphi
\end{pmatrix} ,
\label{supergolden}
\end{align}
this allows us to obtain the most efficient known universal single-qubit gate set \cite{KT-2I}, providing a factor of $\sim 5.9$ complexity saving as compared to the usual Clifford + $T$ gate set \cite{boykin}.

Turning our attention to the 
11-qubit PI code, we see that 
$\bar X_B$ and $\bar Z_B$ implement the logical $X$ and logical $Z$ operators respectively. Application of the transversal $T$ gate gives us a non-Clifford logical gate $\bar Z_B(3\pi/4) = T^{\otimes 11}$.
Thus, we can obtain a logical $T$ gate by applying 
$ Z^{\otimes 11} (T^\dagger)^{\otimes 11}$.
Hence we can obtain a 
logical $T$ gate by applying 
$Z(3\pi/4)^{\otimes 11}$.

There are also PI codes that can implement logical $Z$ rotations using transversal gates with different angles. 
These codes have logical codewords
\begin{align}
|0_{b,g}\> &\coloneqq
( 
\sqrt{2b-g} |D_{0}^{2b+g}\>)
+
\sqrt{2b+g}|D_{2b}^{2b+g}\>)/\sqrt{4b} ,\notag
\\
|1_{b,g}\> &\coloneqq (
\sqrt{2b-g}  |D_{2b+g}^{2b+g}\>
+
\sqrt{2b+g} |D_{g}^{2b+g}\>) /\sqrt{4b},
\end{align}
where $g$ is a positive integer and $2b \ge g+1$. We call such a code a $(b,g)$-PI code.
In Appendix~\ref{app:bg is AAB}, we show that is a special case of the Aydin {\it et al.} code \cite{aydin2023family} that corrects 
1 error if $g \ge 3$ and $2b-g\ge 3$.
For this code, the transversal gate $Z(\pi/b)^{\otimes n}$  
applies a logical $Z$-rotation with angle $\pi g /b$ as shown in Appendix~\ref{app:bg transversal}.

When $g=3$, the $(b,g)$-PI code is equal to the Kubischta-Teixeira PI code 
on $2b+3$ qubits and with distance $3$ \cite{kubischta2023not}. 
Hence we can think of the $(b,g)$-PI code as a generalization of the $(2b+3)$-qubit Kubischta-Teixeira PI code. 
Note that the $(4,3)$-PI code is the 11-qubit PI code.
When $b=2^{r-1}$ for $r \ge 3$, the Kubischta-Teixeira PI code allows one to implement a logical gate in the $r$-th level of the Clifford hierarchy using transversal gates.

For higher distance codes that support logical
$Z(\pi g/b)$ rotations with transversal gates, we introduce 
$(b,g,m)$-PI codes, which generalize
$(b,g)$-PI codes.
These codes use $N = 2bm+g$ qubits.
In general, for any positive integer $m$, we define
\begin{align}
|0_{b,g,m}\>
=
\sum_{k=0}^m
\sqrt{\binom m k} 
\frac{\gamma_{b,g,m,k} }
{2^m \sqrt{(2m-1)!!}}
|D^N_{2kb}\>,
\end{align}
and correspondingly $|1_{b,g,m}\>=X^{\otimes N}|0_{b,g,m}\>$,
where for $k = 0,\dots,m$, we define
\begin{align}
\gamma_{b,g,m,k}
&=b^{-m/2} \prod_{i=k+1}^m \sqrt{2ib-g} \prod_{j=m-k+1}^m \sqrt{2jb+g} .
\end{align}
Note that the 
$(b,g,1)$-PI codes is a $(b,g)$-PI codes. 
We prove in the Appendix that when $g,2b-g \ge 2t+1$ and $m \ge t$, 
a $(b,g,m)$-PI code has distance at least $2t+1$,
and hence corrects $t$ errors. 
These $(b,g,m)$-PI codes are equivalent to the code of Aydin {\em et al.} only when $m=1$. 
For $m\ge 2$, the $(b,g,m)$-PI codes form a new family of PI codes.

For $m=2$, we find generalized bg codes of the form
\begin{align}
|0_{b,g,2}\>
=&
\frac{\sqrt{(2b-g)(4b-g)}}{4b\sqrt{3}} 
 |D^N_{0}\>\notag\\
&+
 \frac{\sqrt{(4b+g)(4b-g)}}{2b\sqrt{6}} 
|D^N_{2b}\>\notag\\
&+
\frac{\sqrt{(2b+g)(4b+g)}}{4b\sqrt{3}} 
|D^N_{4b}\>.
\end{align}
For $m=3$, we find generalized bg codes of the form
\begin{align}
|0_{b,g,3}\>
=&
\frac{\sqrt{(2b-g)(4b-g)(6b-g)}}{8b\sqrt{15b}}
 |D^N_{0}\>\notag\\
&+
 \frac{\sqrt{(4b-g)(6b+g)(6b-g)}}{8b\sqrt{5b}} 
|D^N_{2b}\>\notag\\
&+
\frac{\sqrt{(4b+g)(6b+g)(6b-g)}}{8b\sqrt{5b}}
|D^N_{4b}\>\notag\\
&+
\frac{\sqrt{(2b+g)(4b+g)(6b+g)}}{8b\sqrt{15b}}
|D^N_{6b}\>.
\end{align}
For $m=4$,
we find
\begin{align}
|0_{b,g,4}\>
=&
\frac{\sqrt{(2b-g)(4b-g)(6b-g)(8b-g)}}{16 b^2\sqrt{105}}
 |D^N_{0}\>\notag\\
&+
 \frac{\sqrt{(4b-g)(6b-g)(8b+g)(8b-g)}}{8b^2\sqrt{105}} 
|D^N_{2b}\>\notag\\
&+
\frac{\sqrt{(6b+g)(6b-g)(8b+g)(8b-g)}}{8b\sqrt{70}}
|D^N_{4b}\>\notag\\
&+
\frac{\sqrt{(4b+g)(6b+g)(8b+g)(8b-g)}}{8b\sqrt{105}}
|D^N_{6b}\>\notag\\
&+
\frac{\sqrt{(2b+g)(4b+g)(6b+g)(8b+g)}}{16b\sqrt{105}}
|D^N_{8b}\>,
\end{align}
and $|1_{b,g,4}\>=X^{\otimes N}|0_{b,g,4}\>$.

\subsection{Even-odd quantum codes}
Many PI codes are even-odd quantum codes. 
Indeed for  $g$ odd, the family of $(b,g)$-PI codes introduced above and the family of $(g,n,u)$-PI codes described in Ref. \cite{ouyang2014permutation} are even-odd quantum codes.
We now show that many stabilizer codes share this property as well. Bacon-Shor codes \cite{shorcode-PhysRevA.52.R2493,bacon-PhysRevA.73.012340,baconshor} that are concatenations of odd length codes are even-odd quantum codes.
The logical zero and one of a repetition code are 
$|0\>^{\otimes n}$ 
and 
$|1\>^{\otimes n}$ 
respectively,
which have even and odd weight computational basis states for odd $n$.
Moreover, in the Hadamard basis, the logical zero and one also are superpositions of only 
even and odd weight computational basis states for odd $n$.
Since the concatenation of an even-odd quantum code with an even-odd quantum code results in an even-odd quantum code, the Bacon-Shor codes that are the concatenation of odd length repetition codes in the standard basis and the Hadamard basis must also be even-odd quantum codes.

The 2D-color codes with even stabiliser weights and an odd number of qubits are also even-odd quantum codes.
Such codes have stabilisers generated by Paulis operators with only $X$-type stabilisers and $Z$-type stabilisers, and admit transversal Clifford gates as their logical gates. 
Note that we can write the logical zero of such a code to be proportional to the state $\sum_{P \in S} P|0\>^{\otimes n}$,
where $S$ denotes the stabiliser of the code. 
The $X$-type generators have even  weights, which implies that the logical zero operator must be a superposition over even weight computational basis states.
Since the transversal $X$ gate is the logical $X$ gate which takes the logical zero state to the logical one state and there is an odd number of qubits, 
the logical one state must be a superposition 
over odd weight computational basis states. 
Hence such 2D-color codes are even-odd quantum codes. 

\subsection{High-level code-switching protocol}
\label{sec:high-level code-switching}

One measurement-free approach to swap two qubits is to apply three controlled-NOT (CNOT) gates. The following circuit shows how to swap a pair of qubits with two CNOT gates if one of the states is known \cite[Eq (1)]{ZLC00}:
\begin{align}
\Qcircuit @C=1em @R=.7em {
    \lstick{\ket{\psi }_A} 
    & \ctrl{1} & \targ & \qw 
    &\rstick{\ket{0 }_A}  \\
   \lstick{\ket{0 }_B} &  
   \targ  &  \ctrl{-1}   & \qw  
   &  \rstick{ \ket{\psi }_B}
}\label{2cnot-swap}
\end{align}
For our code-switching protocol, we treat system $A$ as the stabiliser code on $|A|$ qubits, and system $B$ as a PI code on $|B|$ qubits.
Based on \eqref{2cnot-swap}, we can write the logical variant of the swap circuit using two CNOT gates and a known PI ancilla prepared in the $|0_B\>$. 
Using this idea, 
starting from a stabiliser code where we can perform logical Clifford gates transversally, we switch to a PI code to perform a logical non-Clifford gate transversally before switching back to the stabiliser code. 
We can achieve this using the following protocol:
\[
\Qcircuit @C=1em @R=.7em {
    \lstick{\ket{\psi_A}}  & \ctrl{1} & \targ & \qw & \targ & \ctrl{1} & \qw & \rstick{\bar Z_A(\omega')\ket{\psi_A}} \\
   \lstick{\ket{0_B}} &    \targ  &  \ctrl{-1}   & \gate{Z(\omega)^{\otimes |B|}}    & \ctrl{-1} & \targ  & \qw & \rstick{\ket{0_B}} 
}
\]
Using the 11-qubit PI code, we can obtain the logical $T$ gate on the stabiliser code with $\omega = 3\pi/4$ and $\omega' = \pi/4$.
Using the 7-qubit PI code, we can obtain a logical non-Clifford on the stabiliser code with $\omega = 2\pi/5$ and $\omega' = 4\pi/5$.

Since both of these codes have distance equal to 3, they both allow the correction of any single-qubit error on the logical information.

We can code-switch to PI codes  
which admit transversal implementation of other exotic non-Clifford gates. 
Now let $g$ and $b$ be coprime so that there exists an integer $k$ such that ${\rm mod}(k g,b)=1$.
Then the $(b,g)$-PI code on $(2b+g)$ qubits can implement the logical $Z(\pi u /b)$ gate transversally on the stabiliser code with $\omega = u k \pi/b$ for any integer $u$, while allowing the correction of  
of a single error if $g\ge 3$ and $2b-g\ge 3$.

\section{Geometric phase gates}

Our proposed mechanism to process quantum information in PI codes uses geometric phase gates (GPGs) arising from coupling quantum spins to a coherently controllable bosonic mode
\cite{molmer_multiparticle_1999,Alfredo-Luis_2001}. This flavor of GPG is a standard tool for entangling gates in trapped ion quantum processors \cite{Leibfried2003}, where the bosonic mode is a quantized motional mode. Other suitable architectures, as outlined in Ref.~\cite{jandura2023nonlocalmultiqubitquantumgates}, include trapped atoms in optical cavities
\cite{PhysRevLett.91.097905}, superconducting qubits coupled via a driven microwave resonator \cite{PhysRevX.9.041041}, or polar molecules coupled to microwave stripline cavities \cite{André2006}.

To understand the action of the gate, consider the spin operator 
$\hat{w}_{\Gamma}=\frac{|\Gamma|}{2}{\bf 1}-\hat{J}^z_{\Gamma}$
that counts the Hamming weight of spin states,
where 
$\hat{J}^z_{\Gamma}=\frac{1}{2}
\sum_{j=1}^{|\Gamma|} Z_{\Gamma,j}$
 is the collective angular momentum operator.
 Here, $Z_{\Gamma,j}$ applies a phase flip on the $j$th qubit of system $\Gamma$ and applies the identity operator on all other qubits of $\Gamma$.
We describe three types of GPGs here. The first two use an interaction between the spins in $\Gamma$ and the mode that is linear in the creation and annihilation operators of the mode,
while the third
uses a dispersive interaction that is quadratic in bosonic operators:
\begin{equation}
U=\left\{\begin{array}{cc} e^{i\phi \hat{w}^2_{\Gamma}} & \text{Linear\ GPG-A\ \cite{molmer_multiparticle_1999,jandura2023nonlocalmultiqubitquantumgates}}\\
e^{-i I/(\delta-\hat{w}_{\Gamma} g^2/\Delta)} & \text{Linear\ GPG-B \ \cite{jandura2023nonlocalmultiqubitquantumgates}}
\\
e^{-i 2 \chi\sin(\theta \hat{w}_{\Gamma}+\beta)} & \text{Non-linear\ GPG \ \cite{PhysRevA.65.032327}}\end{array}\right.
\end{equation}
In the Linear GPG-A gate, also known as the M\o lmer-S\o rensen gate, the angle $\phi\in[0,2\pi)$ is fully tunable through the control pulses used in the implementation. Recently it has been shown using optimal control methods that this gate together with global rotations of the spins provides for efficient exactly universal state and unitary synthesis in the Dicke state space \cite{Gutman_2024,bond2023efficient,srivastava2024entanglementenhancedquantumsensingoptimal}. In the Linear GPG-B gate introduced in Ref.~\cite{jandura2023nonlocalmultiqubitquantumgates} the relevant parameters are intensity $I$, and the detunings $\delta,\Delta$ which are tunable within a range whereas the coupling strength $g$ between the spins and cavity is usually treated as constant. This gate enables a broader class of entangling gates like the multi-controlled phase gate $C_{N-1}(Z)$. 
For the non-linear GPG, the angles $\chi,\theta,\beta\in[0,2\pi)$ are fully controllable through a combination of dispersive interaction evolutions punctuated by displacement operations on the mode. Complemented by global spin rotations this gate is also exactly universal for state and unitary synthesis in the Dicke space \cite{johnsson2020geometric}.

\begin{figure}[htp]
\includegraphics[width=0.9\columnwidth]{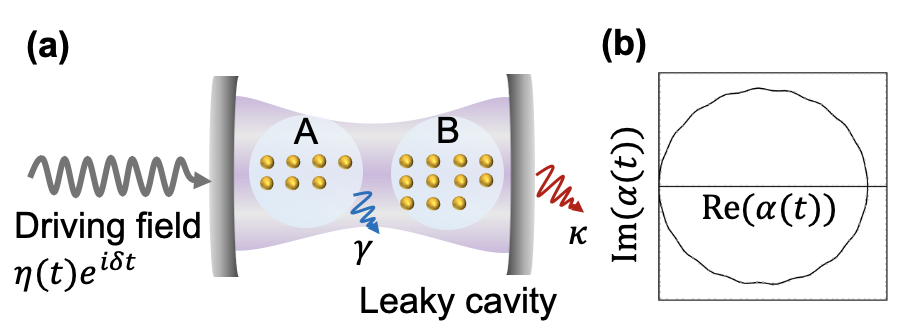}
\caption{(a) Two registers of qubits are coupled to a common cavity mode with strength $g$. The cavity mode decays at a rate $\kappa$, while the spins decay from their excited states at a rate $\gamma$. The cavity acts as a mediator of interactions between the qubits, allowing the two registers to influence one another through their coupling to the same cavity mode. (b) Illustration of a closed trajectory in phase space of a coherent state $\alpha(t)$ of the cavity mode associated to an eigenspace of the spin operator $\hat{w}_{\Gamma}$ for a Linear GPG. After the gate, the mode and the spins are disentangled and the phase accumulated on the eigenspace depends on the area of the trajectory. 
}
\label{fig:sketch}
\end{figure}

\section{PI ancilla state preparation}
\label{stateprepapp}

Logical state preparation of a target PI state $|\Psi_{\rm t}\>$ can be achieved  using the linear GPG-A gate, which we simply denote (l-GPG) \cite{bond2023efficient,srivastava2024entanglementenhancedquantumsensingoptimal}. Starting in the product state $\ket{D_0^N}$, one applies a  sequence of gates:
\begin{align}
\ket{\Psi}=&\prod_{p=1}^{P}[R(\theta_p,\xi_p,\gamma_p) U(\phi_p)]\ket{D_0^N}
\label{eq:prep_linear}
\end{align}
that prepares an output approximating $|\Psi_t\>$,
where $P$ denotes the number of pulses. The GPGs $U(\phi_p)$ are interleaved with global rotations around an arbitrary axis as
\begin{align}
    R(\theta_p,\xi_p,\gamma_p) &= R_z(\theta_p)R_y(\xi_p)R_z(\gamma_p)  \notag\\
                               &=e^{i\theta_p\hat{J}^z}e^{i\xi_p\hat{J}^y}e^{i\gamma_p\hat{J}^z}.
\end{align}
The parameters $\theta_p$, $\xi_p$, $\gamma_p$ and $\phi_p$ are chosen numerically to minimize the infidelity $1-|\langle\Psi\ket{\Psi_{\rm t}}|^2$.
Here, for simplicity, instead of using the definition of the l-GPG above with $U(\phi)=e^{i\phi \hat{w}_{\Gamma}^{2}}$, we exploit the decomposition
\begin{align}
e^{i\phi_p\hat{w}_{\Gamma}^{2}}=e^{i\phi_p(\frac{|\Gamma|^2}{4}-|\Gamma|\hat{J}_{\Gamma}^z+\hat{J}_{\Gamma}^{z^{2}})},  
\end{align}
and instead define $U(\phi_p) = e^{i\phi_p\hat{J}_{\Gamma}^{z^{2}}}$, since the constant and the $\hat{J}_{\Gamma}^z$ term in the exponent can be absorbed into the global rotation $R$ during the optimization. This allows us to treat the contribution by evolution generated by $\hat{J}_{\Gamma}^z$ as part of a uniform rotation without affecting the overall dynamics of the system.

We now discuss the state preparation using l-GPGs in the presence of errors. We assume that the rotations are noise-free, as they can typically be performed fast relative to the spin mode coupling $g$. We apply the error model presented in Ref.~\cite{jandura2023nonlocalmultiqubitquantumgates} in the presence of losses from cavity mode decay at rate $\kappa$ and the spontaneous emission of the spin excited state at rate $\gamma$. In the Dicke subspace, the ideal $p$th l-GPG $U(\phi_p)$ is modified to the erroneous mapping as
\begin{align}
 \mathcal{E}(\phi_p,\rho_{p-1})=\sum_{n,m=0}^{N} \bra{D_n^N}\rho_{p-1}\ket{D_m^N}f_{n,m}(\phi_p)\ket{D^N_n}\bra{D^N_m},
\label{eq:prep_linear_err}
\end{align}
where
\begin{align}
f_{n,m}(\phi_p)=&e^{-(m-n)^2\frac{|\phi_p|}{2}\sqrt{\frac{2(1+2^{-N})}{C}}-(m+n)\frac{|\phi_p|}{2}\frac{1}{\sqrt{2C(1+2^{-N})}}}  e^{-i(n^2-m^2)\phi_p}
\end{align}
and $C=g^2/\kappa\gamma$ is the cooperativity of the cavity supporting the mode. Here, we apply the absolute value on the parameter $\phi_p$ as it can take negative values during numerical optimization. 
Note the map $\mathcal{E}$ is not trace preserving as it treats decay from the excited state as leakage. This provides a lower bound on the process fidelity \cite{jandura2023nonlocalmultiqubitquantumgates}.

The full preparation is a concatenation of erroneous GPGs $\mathcal{E}(\phi_p,\rho_{p-1})$ interleaved with error-free rotations $R(\theta_p,\xi_p,\gamma_p)$, which is written as,
\begin{align}
    \mathcal{E}_{\rm l-GPG}=\mathcal{E}^{P}_{\rm l-GPG} \circ \mathcal{E}^{P-1}_{\rm l-GPG} \circ \cdots \circ \mathcal{E}^{1}_{\rm l-GPG},
\label{eq:full_prep_lGPG}
\end{align}
where $\mathcal{E}^{p}_{\rm l-GPG} = R(\theta_p,\xi_p,\gamma_p) \mathcal{E}(\phi_p,\rho_{p-1}) R^{\dagger}(\theta_p,\xi_p,\gamma_p)$. We apply Eq.~(\ref{eq:full_prep_lGPG}) to the initial state $\rho_0 = \ket{D^N_0}\bra{D^N_0}$ and use infidelity $1-\bra{\Psi_{\rm t}}\rho_{\rm P} \ket{\Psi_{\rm t}}$ as the cost function, where $\rho_{\rm P}$ is the output state after $P$ number of pulses are applied, and $\ket{\Psi_{\rm t}}$ is the target state. In our optimization, we fix the total number of pulses $P$ and utilize MATLAB's built-in toolbox, which employs the Sequential Quadratic Programming (SQP) method, to minimize the cost function and obtain lists of parameters $\theta_p$, $\xi_p$, $\gamma_p$ and $\phi_p$. Note that the optimization method applied is non-deterministic. 
\newline

\section{Implementing entangling logical gates}
The different GPG's discussed above enable use of different logical operations for code switching. 
The non-linear GPG can implement logical CNOT gates between a non even-odd PI code, like the PI-7 code, and an even-odd stabilizer code (see Appendix~\ref{CNOTNLGPG}).
However, the dispersive interactions between the spins and the cavity mode necessary to instantiate the non-linear GPG are not directly accessible in many physical systems \cite{johnsson2020geometric}, and the dispersive interaction strengths tend to be smaller than linear interaction strengths implying slower gates. We therefore focus most of our analysis on using linear GPGs.

\subsection{Using linear GPGs with the PI-11 code}

Instead of the switching circuit shown in Sec.~\ref{sec:high-level code-switching}, we discuss a variation of it, with the idea that logical CZ gates are easier to be implemented than CNOTs. In what follows we show how to implement the following circuit for code switching between the Steane code (system A) and the 11-qubit PI code (system B):
\[
\Qcircuit @C=1em @R=.7em {
    \lstick{\ket{\psi_{A}} }  & \ctrl{1} & \gate{\overline{H}} & \ctrl{1} & \qw & \ctrl{1} & \gate{\overline{H}} & \ctrl{1} & \qw & \rstick{\bar R_z(\omega)\ket{\psi_{A}} } \\
   \lstick{\ket{+_{\rm pi11}} } &    \control \qw  & \gate{\overline{H}}&  \control \qw   & \gate{\bar R_z(\omega)}   & \control \qw & \gate{\overline{H}} & \control \qw & \qw & \rstick{\ket{+_{\rm pi11}}}
   }
\]

 \begin{figure}[htbp]
\includegraphics[width=0.9\columnwidth]{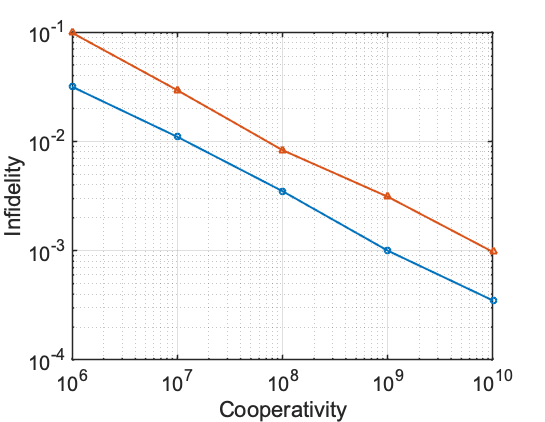}
\caption{Process infidelity $1-F_{\rm pro}(\mathcal{E}_{\overline{H}},\overline{\mathcal{H}})$ for implementing the logical Hadamard gate (orange, triangles) on the PI-11 code, and state infidelity $1-\bra{+_{\rm pi11}}\rho\ket{+_{\rm pi11}}$ for preparing the logical $\ket{+_{\rm pi11}}$ state (blue,circles), using l-GPGs as a function of cooperativity $C$. In both cases $P=10$ gate sequences are used for the state preparation steps.}
\label{fig:infidelity}
\end{figure}
\subsubsection{Implementing the $\rm{C}_{\textsf{A(B)}}{\rm Z}_{\textsf{B(A)}}$ gate }
To realize the $\rm{C}{\rm Z}$, which acts symmetrically on the control and target, it suffices to construct a gate that up to global phase on the logical subspace, applies a minus sign on the logical basis state $\ket{1}_A\ket{1}_B$ and acts trivially on the other logical basis states. 
When both $A$ and $B$ are even-odd quantum codes, as is the case here, then we have the following decomposition in terms of three l-GPGs
\begin{equation}
\rm{C}{\rm Z}=e^{i \frac{\pi}{2} \hat{w}^2_{A\cup B}} 
e^{-i \frac{\pi}{2} \hat{w}^2_{A}}e^{-i \frac{\pi}{2} \hat{w}^2_{B}}. 
\label{CZinGPG}
\end{equation}
Here an l-GPG involving the operator $\hat{w}_{A\cup B}=\hat{w}_{A}+\hat{w}_{B}$ can be obtained by coupling all the spins constituting codes $A$ and $B$ to the same bosonic mode. The process infidelity for performing each of the three l-GPGs in this sequence is \cite{jandura2023nonlocalmultiqubitquantumgates}
\begin{equation}
    1-F=\frac{\pi|\Gamma|}{2\sqrt{2(1+2^{-|\Gamma|})C}}.
\end{equation}.

\subsubsection{Implementing the logical $\overline{\rm{H}}$ gate}
Next, we discuss a way to implement a logical Hadamard gate $\overline{H}$ for the PI-11 code.
In the code space, the eigenvalues of $\overline{H}$ are $\pm 1$ with corresponding eigenstates:
\begin{align}
     |\lambda_{+}\rangle &= \frac{1}{\sqrt{2(2+\sqrt{2})}} [(1+\sqrt{2})|0_{\rm pi11}\rangle+|1_{\rm pi11}\rangle ],  \notag \\
     |\lambda_{-}\rangle &= \frac{1}{\sqrt{2(2-\sqrt{2})}} [(1-\sqrt{2})|0_{\rm pi11}\rangle+|1_{\rm pi11}\rangle ].  
 \label{eq:H_eigenvector}
 \end{align}

We can use the method of PI state preparation to construct a unitary on the entire Dicke space which acts correctly in the logical subspace. Using the spectral decomposition, $\overline{H}$ must have the form,
\begin{align}
    \overline{H} = |\lambda_{+}\rangle\langle\lambda_{+}| -|\lambda_{-}\rangle\langle\lambda_{-}|+\sum_s e^{i\beta_s}\ket{\beta_s}\bra{\beta_s},
 \label{eq:logical_H}
 \end{align}
here $\{\ket{\beta_s}\}$ are any set of orthonormal vectors which are contained completely in the ortho-complement to the code space, and $\{\beta_s\}$ are any eigenvalues of on those states. 
We can construct this gate with linear GPGs using the method of state preparations adapted for subspaces \cite{PhysRevA.71.052318}. Specifically we can write
\begin{align}
    \overline{H} = W e^{i \pi \ket{D^N_N}\bra{D^N_N}}W^{\dagger}
    \label{eq:Hdecomp}
 \end{align}
where $W$ is any unitary extension of the state preparation: $W\ket{D^N_N}=\ket{\lambda_-}$. The diagonal phase gate $e^{i \pi \ket{D^N_N}\bra{D^N_N}}$ is the $C_{N-1}(Z)$ gate, where $N=11$ for the PI-11 code. This gate can be synthesized from linear GPG-B gates using the methods in Ref.~\cite{jandura2023nonlocalmultiqubitquantumgates}.

We characterize the performance of the implemented Hadamard gate described by Eq.~(\ref{eq:Hdecomp}) in the presence of losses using the fidelity-based measure $F_{\rm pro}(\mathcal{E}_{\overline{H}},\overline{\mathcal{H}})$ proposed in Ref.~\cite{gilchrist2005distance} to compare the performance of actual realized process $\mathcal{E}_{\overline{H}}$ against the ideal unitary process $\overline{\mathcal{H}}(\rho)=\overline{H}_{\rm ide}\rho\overline{H}_{\rm ide}^{\dagger}$, where $\overline{H}_{\rm ide}$ is the ideal Hadamard gate in logical subspace. The actual process for the implementation of logical Hadamard gate is given by
\begin{align}
    \mathcal{E}_{\overline{H}}=\mathcal{E}_{\rm l-GPG}\circ\mathcal{E}_{\rm ph}\circ\mathcal{E}^{\rm R}_{\rm l-GPG},
\label{eq:sketch_map_H}
\end{align}
where $\mathcal{E}_{\rm ph}(\rho)=F_{\rm ph}C_{N-1}(Z)\rho C^{\dagger}_{N-1}(Z)$ is the erroneous map of the middle phase gate in Eq.~(\ref{eq:Hdecomp}). Here, $F_{\rm ph}$ denotes the fidelity of the multi-qubit phase gate. \YJ{This multi-qubit phase gate is implemented in the weak-drive regime, where an adiabatic evolution of the joint cavity-qubit system is exploited to realize the operation. The dominant noise arises from cavity decay $\kappa$ and spontaneous emission from the excited state $\gamma$, which are incorporated through the cooperativity $C$. Numerical simulations reported in Ref.~\cite{jandura2023nonlocalmultiqubitquantumgates} (Sec. IV B, Fig.~3(b,c)) show that these error channels lead to an infidelity scaling of approximately $1-F_{\rm ph} \sim 1.8 \times N/\sqrt{C}$, which we quote here for completeness.}
$\mathcal{E}_{\rm l-GPG}$ ($\mathcal{E}^{\rm R}_{\rm l-GPG}$) is the (reverse of) preparation mapping associated with $W$ ($W^{\dagger}$) in the presence of losses, as detailed in Eq.~(\ref{eq:full_prep_lGPG}). Here, for simplicity, we optimize parameters solely for the preparation process and apply the same parameters to the reverse mapping without additional optimization. The fidelity measure $F_{\rm pro}(\mathcal{E}_{\overline{H}},\overline{\mathcal{H}})$ is computed based on process matrices, which necessitates the superoperator formalism for each mapping in Eq.~(\ref{eq:sketch_map_H}). The detailed calculations are provided in Appendix~\ref{app:fidelity_measure}.

In Fig.~\ref{fig:infidelity}, we plot the infidelity for: the implementation error of the Hadamard gate (orange), and the preparation error (blue) for the PI-11 code as a function of the cooperativity $C$. The simulation is performed in the presence of loss errors as described in Eq.~(\ref{eq:prep_linear_err}). Following the approximation in Ref.~\cite{jandura2023nonlocalmultiqubitquantumgates} (Eq. (25)), the infidelity of the $\overline{H}$ implementation can be estimated as $1-F_{\rm pro}(\mathcal{E}_{\overline{H}},\overline{\mathcal{H}}) \sim 21.78 \times N/\sqrt{C}$, where we consider $18$ non-trivial GPG operations for the preparation step and its inverse, with $\theta=\pi/2$.
Similarly, the infidelity for state preparation is estimated by $1-\bra{+_{\rm pi11}}\rho\ket{+_{\rm pi11}} \sim 9.99 \times N/\sqrt{C}$, where $9$ non-trivial GPG gates are considered, again with $\theta=\pi/2$. Our simulation results indicate a somewhat better performance with the infidelity scaling as $1-F_{\rm pro}(\mathcal{E}_{\overline{H}},\overline{\mathcal{H}}) \sim 8.29 \times N/C^{0.4985}$ and $1-\bra{+_{\rm pi11}}\rho\ket{+_{\rm pi11}} \sim 2.80 \times N/C^{0.4953}$, indicating lower infidelity compared to the approximation.

\section{Approximation of the super golden gate}

The most efficient single-qubit gate set is given by representations of the binary icosahedral group $2I$ (which can be implemented with the Pollatsek-Ruskai code) and the super golden gate $\tau_{60}$ given in \eqref{supergolden}. This efficiency is measured in terms of gate counts.
Here we give an approximate decomposition of $\tau_{60}$ which can be achieved by code-switching between bg-PI codes, the Pollatsek-Ruskai code that implements 2I, and PI codes that implement the transversal $T$ gate. 

Let $S=Z(\pi/2)$, and $T=Z(\pi/4)$. Let $H$ denote the Hadamard gate, and $Z$ denote $Z(\pi)$.
Noting that the super golden gate squares to the identity, we can write it in an Euler decomposition as 
\[
\tau_{60}=
 i 
e^{-i\frac{\pi}{8}Z}e^{-i\frac{\theta}{2} Y}e^{-i \frac{3\pi}{8}Z},
\]
where $\theta=2\cos^{-1}{\frac{2+\varphi}{\sqrt{5\varphi +7}}}$. We can further simplify
\[
\begin{array}{lll}
\tau_{60}&=&Te^{-i\frac{\theta}{2} Y}ZT^{\dagger}\\
&=&
e^{-i \theta/2} 
TSH Z(\theta) HS^{\dagger}ZT^{\dagger}.
\end{array}
\]
Now $\theta/\pi$ can be approximated by a rational. Defining
\[
 \begin{array}{lll}
 \tilde{\tau}_{60}(\gamma)&=&TSH Z(\gamma)HS^{\dagger}ZT^{\dagger}
  \\
  &=&TF^{\dagger} Z(\gamma)F Z T^{\dagger}
  \end{array}
\]
we find $\|\tilde{\tau}_{60}(\gamma)-\tau_{60}\|< 10^{-6}$
for $\gamma=\frac{\pi 167}{704}$,
where $
\| \cdot \|
 $ denotes the operator norm minimized over global phases.
This could be achieved using: a $(704,167)-$PI code for the small angle $Z$ rotation, a smaller $(b,g)-$PI code for transversal $T$ gates, and a PI code like the Pollatsek-Ruskai code admitting transversal gates in the group $2$I (noting $F^2=F^{\dagger}$ up to a global phase).

\section{On the fault-tolerance of our scheme}

\noindent
{\bf Spin system errors}:-
An operator $Q$ is fault-tolerant on the spin-system for QEC codes on system $A$ and $B$ that correct $t$ errors if for any multiqubit Pauli operator $P$ of weight $t$, there exists an operator $A$ with weight at most $t$ such that $Q P = A Q$. Operationally, this means that when a weight $t$ error $P$ propagates across a quantum circuit $Q$, the error $P$ evolves in the Heisenberg picture to an error $A$ that does not have too many errors. This definition of fault-tolerance on the spin-system is consistent with the definition of fault tolerance proposed in Ref~\cite{AGP05}. For us, the 7 qubit and 11 qubit PI codes correct a single error, and hence $t=1$.  

Each constituent gate in our GPGs and conditional GPGs is spin fault-tolerant, which demonstrates the spin fault-tolerance of our protocol.
These constituent gates are the rotation operators 
$R(\theta \hat J^z_\Gamma)$ and $R(\theta \hat J^y_\Gamma)$ and the mode displacement operator. 
Mode displacement operators are trivially 
spin fault-tolerant because they commute with all spin errors.

To demonstrate the spin-fault tolerance of the rotation operators, it suffices to 
focus our attention on $R(\theta \hat J^z_\Gamma)$.
The rotation operator $R(\theta \hat J^z_\Gamma)$ 
is trivially spin fault-tolerant with respect to $Z$ errors because it commutes with $Z$ errors on the spin system. 
From Ref.~\cite{ouyang2021avoiding}, 
 the rotation operator also obeys the commutation relations 
$
R(\theta \hat J^z_\Gamma) X_{\Gamma,k} 
=
e^{i \theta Z_{\Gamma, k}} X_{\Gamma,k}
R(\theta \hat J^z_\Gamma)   $
and
$
R(\theta \hat J^z_\Gamma) Y_{\Gamma,k} 
=
e^{i \theta Z_{\Gamma, k}} Y_{\Gamma,k}
R(\theta \hat J^z_\Gamma) $ for $X$ and $Y$ errors respectively.
These commutation relations imply that both an $X$ and $Y$ error on the $k$th qubit 
evolves under the rotation operator into an error that remains localized on the $k$th qubit.
Hence the rotation operator $R(\theta \hat J^z_\Gamma)$ is indeed spin fault-tolerant. 
\newline

  \noindent
{\bf Mode errors}:-
Loss or gain errors can occur on our bosonic mode, which we model using the lowering operator $\hat a$ and the raising operator $\hat a^\dagger$ respectively.
We also consider phase errors on the bosonic mode, which are generated by the number operator $\hat a^\dagger \hat a$.
We say that an operator $Q$ is fault-tolerant on the bosonic system, or {Bose} fault-tolerant, if for any operator $\bar B$ that is either a ladder operator or a number operator, there exists an operator $A$ with weight at most $t$ on the spin systems such that $Q\bar B = AQ$.


The rotation operator $R(\theta \hat w_{\Gamma})$
obeys the commutation relations   
$R(\theta \hat w_{\Gamma}) \hat a
=
e^{-i\theta \hat w_{\Gamma}} 
\hat a R(\theta \hat w_{\Gamma}) $
and 
$ R(\theta \hat w_{\Gamma}) \hat a^\dagger
=
e^{i\theta \hat w_{\Gamma}} \hat a^\dagger R(\theta \hat w_{\Gamma}) $.
This shows that ladder operator errors on the bosonic mode can propagate to the spin system as unitaries. 
Hence the rotation operator is not {Bose} fault-tolerant with respect to ladder operators, unless $\theta$ is very close to an integer multiple of $2\pi.$
On the other hand the rotation operator is trivially {Bose} fault-tolerant with respect to the number operator.

While the mode displacement operator does not commute with the ladder operators, it does not introduce additional errors to the spin system because ladder operator errors still propagate to ladder operator errors. Hence the displacement operator is {Bose} fault-tolerant with respect to ladder operators.
The displacement operator is furthermore {Bose} fault-tolerant with respect to number operators because it propagates the number operator to a linear combination of the number operator and ladder operators. 
Hence the displacement operator is {Bose} fault-tolerant.

To mitigate the impact of the rotation gate not being {Bose} fault-tolerant in general,
we should aim to perform both the displacement gate and the rotation gate as quickly and as accurately as possible, minimizing the number of gain or loss errors. On the other hand, our protocol is fault-tolerant with respect to bosonic phase errors.

\GKB{
The dominant error of our scheme arises from finite cooperativity of the cavity. Recent experimental work \cite{zhang2025opticallyaccessiblehighfinessemillimeterwave} demonstrates precision control with a high finesse Fabry–P\'erot cavity which provides for a single particle cooperativity
of $C = 6.75\times 10^5$ utilizing circular Rydberg transitions coupled to a microwave transition. Higher cooperativities of $10^8-10^9$ are achievable \cite{Kuhr:2006} using millimeter wave cavities, though simultaneously achieving optical access is challenging. An alternative to using photonic modes is to use phononic modes, like a motional mode of an ion trap array. Variational quantum circuits using the same gates as our scheme here have already been demonstrated \cite{Marciniak2022} for the preparation of squeezed quantum states for quantum metrology. For this architecture, the mode is very long-lived and the cooperativity is not the relevant parameter; rather the error associated with coupling to the motional mode is dominated by laser fluctuation noise. The analysis in Ref.~\cite{bond2023efficient} shows that in the presence of laser noise that induces a fluctuation in the phase $\phi$ of the linear GPG-A by an amount $\delta\phi=0.01\%$, the full state preparation infidelity of a random state in the Dicke subspace of an $N=10$ qubit system is below $10^{-6}$.} \YJ{If we consider larger codes, such as the $(6,5,2)-$PI code, referred to as the PI$-29$ code and capable of correcting two errors, the same reference reports that the corresponding error can remain below $10^{-4}$ for $N=30$.}

\GKB{
\section{Performance}
To highlight the advantages of our measurement-free code switching scheme relative to others based on measuring stabilizers, we can compare the control complexity required. For simplicity, we consider switching between two even-odd codes, one code $A$ having $N_A$ qubits and admitting transversal Clifford gates, and a code $B$ with $N_B$ qubits admitting a transversal non-Clifford gate.}

\GKB{
The gate cost, counted by number of transversal gates or cavity pulses, to switch into a PI code $B$ using our method is as follows. State preparation requires at most $2N_B$ global pulses~\cite{bond2023efficient}. The logical Hadamards can be done as one transversal gate on the stabilizer code, and on the PI code via \eqref{eq:Hdecomp} $2N_B$ gates for the state preparation unitary $W$, $2N_B$ gates for $W^{\dagger}$, and $N_B-1$ gates for the $C_{N_B-1}(Z)$ gate \cite{jandura2023nonlocalmultiqubitquantumgates}. Finally the logical CZ gate is done with $3$ cavity pulses and there are two of them. This gives a total count of $7N_B+6$ gates for switching into the PI code. 
\yo{The total round-trip switching cost, with the transversal implementation of the non-Clifford gate is then
at most $14N_B+13$.}
Two noteworthy points are that the gate count is independent of $N_A$ and the native nonlocality of the cavity-mediated entangling gates means no qubit swapping or shuttling is needed at any point. Additionally, only minimal addressability is required; namely it is only needed to distinguish qubits in one code versus the other.}

\GKB{
The cost to switch from one stabilizer code $A$ with some transversal Clifford gates to another stabilizer code $B$ with a transversal non-Clifford gate can be fairly compared to our scheme by considering a non-fault-tolerant set of stabilizer measurements: first on the stabilizers for $B$ to prepare a logical state, and second for parity checks on $A\cup B$. 
Following the construction of Ref.~\cite{jain2025transversalcliffordtgatecodes}, we choose for code $A$ a doubly-even CSS code that includes the Hadamard among its set of transversal logical gates and for code $B$ a triorthogonal code that supports a transversal logical $T$ gate.}
\yo{The double-even CSS codes in Ref.~\cite{jain2025transversalcliffordtgatecodes} are even-odd codes, because the classical codes used in the CSS construction have even weights, and any doubly-even CSS code admits a strongly transversal logical $X$ and has odd length.}

\yo{We analyze the gate complexity of code-switching from code A to code B, where code B is either a triorthogonal code of distance $d$ or our $(4,3,(d-1)/2)$-PI code. Note that such PI codes support a transversal logical $T$ gate and are furthermore even-odd codes because the parameter $g$ in the $(b,g,m)$-PI code is odd. 
We supply a lower bound and an upper bound for the gate count of using a triorthogonal code and a PI code for code B respectively. 
 For the lower bound, the gate count is at least $w(N_B-1)$ considering non-fault-tolerant stabilizer measurements,
 where $w$ is the minimum weight of the triorthogonal code's checks. 
 Table~\ref{tab:roundtrip_costs} shows the shortest triorthogonal codes for a given distance $d=3,5,7,9,11$. 
 For these triorthogonal codes, we prove that $w \ge 8$. For $d=3,5,7$, the triothogonal codes here are triply even, and hence $w\ge 8$. For the higher distance codes here, the average weight of the triorthogonal code of distance $d$ is at least the minimum of the weights of the triorthogonal codes here of distance $d-2$, the length of the triorthogonal code of distance $d-2$, 
 and twice the distance of the self-dual classical code used in the doubling construction of Ref.~\cite{jain2025transversalcliffordtgatecodes}.
 By induction, this minimum is at least 8 for all odd $d\ge 9$.
Therefore, the gate count of using triorthogonal codes according to the doubling construction of Ref.~\cite{jain2025transversalcliffordtgatecodes} in Table~\ref{tab:roundtrip_costs} 
 is at least $8(N_B-1)$ considering non-fault-tolerant stabilizer measurements. 
The simple lower bounds we supply apply not only to both measurement-based and measurement-free code-switching, but also to magic-state distillation techniques. 
The gate count would be even higher if we consider (1) spatial locality constraints on the implementation for instance in 2D, which would impose an additional $O(\sqrt N_B)$ overhead, and (2) the fact that the weights of the stabilizer checks grow as $O(\sqrt N_B)$.  
The gate count for our scheme, on the other hand, is upper bounded by $7N_B+6$. A key benefit of PI codes therefore is most pronounced when their lengths are considerably shorter than their corresponding triorthogonal codes. This is evidenced by the inequality 
$7N_{{\rm PI}, B} + 6 < 8 (N_{{\rm tri}, B}-1)$,
for all length $N_{{\rm tri}, B}$ triorthogonal and length $N_{{\rm PI}, B}$ PI codes of distances $d=3,5,7,9,11$ involved.}

\begin{table*}[htbp]
\centering
\begin{threeparttable}
\begin{tabular}{@{}c|l|l|r|r|@{}}
\toprule
\textbf{Distance $d$} & \textbf{Stab. Code A (Clifford)} & 
\textbf{Stab. Code B} & 
\textbf{PI Code B} & 
\textbf{Gate Count (A$\to$B)} \\
\midrule
3  & $[[7,1,3]]$   & $[[15,1,3]]$   
& 
PI-(4,3,1)=((11,2,3)) & 
$\ge 112$, $\le 83$\\
5  & 
$[[17,1,5]]$     & 
$[[49,1,5]]$     & 
PI-(4,3,2) = ((19,2,5))  & 
$\ge 384$, $\le 139$
\\
7  & 
$[[23,1,7]]$      &  $[[95,1,7]]$       & 
PI-(4,3,3) = ((27,2,7))
& $\ge 752$, $\le 195$ 
\\
9 & 
$[[45,1,9]]$    & 
$[[185,1,9]]$  
& 
PI-(4,3,4) = ((35,2,9))
& 
$\ge 1472$, $\le 251$
 \\
11 & 
$[[47,1,11]]$    &  $[[279,1,11]]$     & 
PI-(4,3,5) = ((43,2,11))
& $\ge 2224$, $\le 307$\\
\bottomrule
\end{tabular}
\caption{{\bf Comparison of gate costs for switching from code A to code B, where code B is either a triply even stabilizer code or a PI code.} The stabilizer codes of type A and B are respectively doubly even and triorthogonal codes given in \cite[Table 1,2]{jain2025transversalcliffordtgatecodes}. We calculate a gate lower bound and upper bound for switching from code A to stabilizer code B and a PI-$(4,3,m)$ code B respectively.
Here, code A supports transversal logical Cliffords, and code B supports transversal logical $T$ gates. All codes involved here encode a single logical qubit.
\label{tab:roundtrip_costs}
}
\end{threeparttable}
\end{table*}

\section{Conclusion}

We have presented a method for code switching between a stabilizer code and a permutation invariant code, and along the way have introduced a family of PI codes with a tunable set of transversal, non-Clifford gates. It is noteworthy that our measurement-free code-switching scheme is also applicable to transitions between two stabiliser codes, provided they satisfy the even-odd quantum code structure. In comparison to the measurement-based approach, our scheme eliminates the overhead associated with measurements and reduces the gate operation complexity by utilising collective gate implementations. 
Given the recent development on how quantum error correction on arbitrary permutation-invariant codes can be implemented using geometric phase gates \cite{ouyang2023quantum}, we expect it to be possible to extend our work to a fully fault-tolerant setting, using ideas from fault-tolerance theory \cite{AGP05,gottesman2014fault,chao2018flag,Chao2020.anyflag}.
Additionally, a numerical study of the efficiency of our protocol for fault-tolerant quantum computation is an interesting line of future line of enquiry that is currently beyond the scope of our paper.

\section{Acknowledgements}\label{eq:acknow}

Y.O. acknowledges support from EPSRC (Grant No. EP/W028115/1) and also the EPSRC funded QCI3 Hub under Grant No. EP/Z53318X/1.
G.K.B. and Y.J. acknowledge support from  the Australian Research Council Centre of Excellence for Engineered Quantum Systems (Grant No. CE 170100009).

\bibliography{ref}

\appendix

\section{AAB+ code}
\label{app:AAB}

The AAB+ type code is specified by three integer parameters $g,m$ and $\delta$, with logical codewords given by 
\begin{align}
    |0_{(g,m,\delta,+)}\> &= 
    \sum_{l\ {\rm even} } \gamma b_l |D_{gl}^n\> +  
    \sum_{l\ {\rm odd} } \gamma b_l |D_{n-gl}^n\> ,
    \\
    |1_{(g,m,\delta,+)}\> &= 
    \sum_{l\ {\rm even} } \gamma b_l |D_{n-gl}^n\> +  
    \sum_{l\ {\rm odd} } \gamma b_l |D_{gl}^n\> ,
\end{align}
where $b_l = \sqrt{\binom ml / \binom {n/g-l}{m+1}}$ and $\gamma = \sqrt{ \binom{n/(2g)}{m} \frac{n-2gm}{g(m+1)} }  $, and $n = 2gm+\delta + 1$.
The code corrects $t$ errors when $g \ge 2t+1$, $m \ge t$ and $\delta \ge 2t$.

In Appendix~\ref{app:bg is AAB}, we show that the $(b,g)$-PI code is the $(g,1,2b-g-1)$ AAB+ code. 
From the property of the AAB+ code, the $(b,g)$-PI code  has distance 3 if $g\ge 3$ and $2b-g\ge 2$. 

Next we turn our attention to the transversal non-Clifford gate.
We consider the gate $T_b = Z(\pi/b)$.
Then we see that 
\begin{align}
    T_b^{\otimes n}|0_{b,g}\> & = |0_{b,g}\>, \\
    T_b^{\otimes n}|1_{b,g}\> & = e^{ i g \pi/b}|1_{b,g}\>.
\end{align}
From this, we see that the transversal gate $T_b^{\otimes n}$ gate 
applies a logical $Z$-rotation with angle $\pi g /b$.

\section{Proof that the $(b,g)$-PI code is an AAB+ code}
\label{app:bg is AAB}

To have the $(b,g)$-PI code, the parameter $m$ in the AAB+ type code must be equal to 1, because each logical codeword is a superposition of two Dicke states. 
By setting $m=1$, we have
\begin{align}
b_0 &=1 / \sqrt{\binom{n/g}{2}} = g /\sqrt{n(n-g)/2} ,\\
b_1 &= 1/ \sqrt{\binom{n/g-1}{2}} =g/  \sqrt{(n-g)(n-2g)/2} ,\\
\gamma &= \sqrt{ \frac{n}{2g} \frac{\delta+1}{2g} } 
=\sqrt{ n (\delta+1) }/(2g).
\end{align}
Then we have 
\begin{align}
\gamma b_0 &= \sqrt{ n (\delta+1) }/ \sqrt{2n(n-g)}  = \sqrt{\frac{\delta+1}{2(n-g)}}, \\
\gamma b_1 &= \sqrt{ n (\delta+1) }/ \sqrt{2(n-g)(n-2g)} =\sqrt{\frac{n(\delta+1)}{2(n-g)(n-2g)}}.
\end{align}
Now consider having $n = 2b+g$ for some positive integer $b$ so that $\delta = 2b-g-1$, $n-g = 2b,n-2g= 2b-g$. Then 
\begin{align}
\gamma b_0 &=   \sqrt{\frac{2b-g}{4b}} = \sqrt{ \frac{1}{2} - \frac{g}{4b} }, \\
\gamma b_1 &= \sqrt{\frac{n(2b-g)}{4b(2b-g)}} =  \sqrt{\frac{2b+g }{4b }} =\sqrt{ \frac{1}{2} + \frac{g}{4b} }.
\end{align}
This shows that the $(b,g)$-PI codes are just $(g,1,2b-g-1)$-AAB+ codes.
The condition for $(b,g)$-PI codes to correct a single error is that $g\ge 3$ and $2b-g\ge 3$. When $g=3,m=1,\delta = 2^r-4$, the AAB+ code can implement the logical gate for $T^{2^{3-r}}$ transversally for $r \ge 3$ \cite{aydin2023family}.
   
\section{Transversal non-Clifford logical operations}
\label{app:bg transversal}

Our code supports transversal $X$ as the logical $X$. 
If $n$ is odd (which is when $g$ is odd), then the transversal $Z$ is the logical $Z$.
If $n$ is even (when $g$ is even), the transversal $Z$ gate stabilizes the code.
\newline

\noindent
{\bf Whenever $b$ and $g$ are coprime, we can implement the logical $\pi k/b$ rotation transversally for any $k=0,\dots,b-1$.}
For the $(b,g)$-PI code to correct $1$ error, it suffices to have $2b-g \ge 3$ and $g \ge 3$. 
When $b$ and $g$ are coprime, we can implement a logical $T_b^k$ gate transversally for any $k=0,\dots, b-1$ by setting $g$ as the smallest coprime number to a fixed $b$.  \newline

\section{Construction of codes from linear programming}

Let us consider the problem of constructing quantum codes encoding a single logical qubit that correct against a set of Kraus operators $\{K_i\}$ given the promise that the orthogonality of the Knill-Laflamme quantum error correction conditions are satisfied.

Now suppose that the codespace resides in a subspace of the space spanned by the orthonormal vectors 
$\{ |c_j\> \}$,
and we write the logical codewords as
linear combinations of these orthonormal vectors. 
Specifically, we demand that the logical 
codewords are linear combinations of distinct basis states given by
\begin{align}
|0_L \> 
 &=  \sum_{j \ {\rm even}} \sqrt{x_j}|c_j \> \\
|1_L \> 
 &=  \sum_{j \ {\rm odd}} \sqrt{x_j}|c_j \> ,
\end{align}
where the coefficients $x_j$ are real and non-negative.

Such a method has been described in \cite{movassagh2020constructing}, and we explain in the next section how to apply this idea explicitly to construct extended bg codes that correct errors of increasing distance, and hence ensure the scalability of the codes.

  \section{$(b,g,m)$-PI codes}



Here, we prove that $(b,g,m)$-PI codes have 
distance $d \ge 2t+1$ whenever
$g, 2b-g \ge 2t+1$ and when $m\ge t$.
Clearly, when $g,|2b-g| \ge d$, the code has a bit-flip distance of $d$.
Hence, for the code to have a distance of $d$, 
it suffices to 
to ascertain that the non-deformation conditions of the Knill-Laflamme conditions holds for the diagonal errors, that is
\begin{align}
\<0_{b,g,m}| P_j  |0_{b,g,m} \>
=
\<1_{b,g,m}| P_j   |1_{b,g,m} \>,
\end{align}
 for all $k=0,\dots,d-1$,
 where $P_j =  Z^{\otimes j} \otimes I^{\otimes N-j}$.
 Now we consider a code 
 with logical codewords
 \begin{align}
  |0_{b,g,m} \> &= \sum_{k=0}^m
  a_k |D^N_{2bk}\>\notag\\
  |1_{b,g,m} \> &= \sum_{k=0}^m
  b_k |D^N_{2bk+g}\>,
 \end{align}
 for real coefficients $a_k,b_k$,
 with the idea of solving for $a_k$ and $b_k$ using linear algebra techniques, akin to the idea in Refs \cite{ouyang2019permutation,movassagh2020constructing}.
Then we can write the non-deformation conditions as
\begin{align}
\sum_{\substack{
0\le k \le m\\ 
}} 
a_k^2 \<D_{2kb}^{N}| P_k |D_{2kb}^{N}\>
=
\sum_{\substack{
0\le k \le m\\ 
}} 
b_k^2 \<D_{g+2kb}^{N}| P_k |D_{g+2kb}^{N}\>,
\end{align}
which is a system of linear equations in the variables $a_k^2$ and $b_k^2$.
The Dicke inner products $\<D^N_w|P_k|D^N_w\>$ can 
be written as Krawtchouk polynomials \cite[Lemma 6]{ouyang2019robust}, 
namely
\begin{align}
\<D^N_{w}|P_k|D^N_w\>
=
K^{N}_{w}(k) / \binom N w,
\end{align}
where
\begin{align}
K^N_k(z) =  \sum_{j = 0}^z \binom {z} {j} \binom{N-z}{k-j} (-1)^j\label{eq:kpoly summation}
\end{align}
is a binary Krawtchouk polynomial.
In the language of generating functions, 
\begin{align}
    K^N_k(z) = [x^k](1-x)^{z}(1+x)^{N-z},
\end{align}
where $[x^k]f(x)$ denotes the coefficient of $x^k$ of a polynomial $f(x)$. 
From the above generating function, we can also see that 
\begin{align}
    K^N_{k}(z) = [y^{N-k}](y-1)^{z}(y+1)^{N-z},
\end{align}
and hence we have the symmetry relation
\begin{align}
    K^N_{N-k}(z) 
    &= [x^{N-k}](1-x)^{z}(1+x)^{N-z}\notag\\
    &= [y^{N-k}](-1)^{z}
       (y-1)^{z}(y+1)^{N-z}\notag\\
    &= (-1)^{z}
    [y^{k}](y-1)^{z}(y+1)^{N-z}
    \notag\\
    &= (-1)^{z} K^N_k(z).\label{k-sym}
\end{align}
We can again rewrite this as a system of homogeneous linear equations, $A{\bf x}=0$, where
\begin{align}
    A  &= 
    \sum_{j=0}^{d-1}
    \sum_{\substack{
0\le k \le m\\ 
}} 
    |j\>  \<2k|
    \<D_{2kb}^{N}| P_j |D_{2kb}^{N}\>\notag\\
    &\quad +
     \sum_{j=0}^{d-1}
    \sum_{\substack{
0 \le k \le m\\ 
}} 
    |j\>  \<2k+1|
    \<D_{g+2kb}^{N}| P_j |D_{g+2kb}^{N}\>,
    \notag\\
    &= 
    \sum_{j=0}^{d-1}
    \sum_{\substack{
0\le k \le m\\ 
}} 
    |j\>  \<2k|
    K_{2kb}^{N}(j) / \binom N {2kb}
    \notag\\
    &\quad +
     \sum_{j=0}^{d-1}
    \sum_{\substack{
0 \le k \le m\\ 
}} 
    |j\>  \<2k+1|
    K_{g+2kb}^{N}(j) / \binom N {g+2kb},
    \label{A-matrix}
\end{align}
and where ${\bf x}  = (x_0,\dots, x_{2m+1})$.
Now define the length $2(m+1)$ vector of ones as ${\bf e}$.
Define the length $2(m+1)$ vector ${\bf f} = (f_0,\dots, f_{2m+1})$ to be such that $f_j = (-1)^j$ for all $j=0,\dots,m$.
The equation \eqref{A-matrix} 
shows that we can find the coefficients $a_k$
by solving the linear program
\begin{maxi}
    {{\bf x}}{ 0 }
    {}{}
    \addConstraint{A {\bf x}}{=0}
    \addConstraint{{\bf f}^T {\bf x}}{=2}
    \addConstraint{(-1)^k x_k \geq 0, \quad k=0}{,\dots,2m+1}
    \label{linear-program}
\end{maxi}
Exploiting the structure of the matrix $A$, we can derive the form of the vector ${\bf x}$ that satisfies the constraints in the above linear program. 
Namely, when
$(-1)^k x_k \geq 0$,
we can set the coefficients of the logical codewords of the generalized bg code as
\begin{align}
    a_k &= \sqrt{ x_{2k} }  \quad {k\ {\rm even}}\\
    b_k &= \sqrt{ -x_{2k+1} }  \quad {k\ {\rm odd}}.
\end{align}
We choose $m$ as the smallest integer for which \eqref{linear-program} has a solution.

Next, let us simplify the linear program \eqref{linear-program}.
Consider a vector 
${\bf y} = (y_0,y_1,\dots, y_{2(m+1)-1})$ such that 
\begin{align}
y_{k} = -y_{2m+1-k}
\end{align}
for all $k=0,\dots, 2m+1$.
Then we can write ${\bf y}$
as a vector that depends on $(m+1)$ variables, in the sense that
\begin{align}
{\bf y}  
= 
(y_0, -y_{m}, 
y_1, -y_{m-1},
\dots,
y_m, -y_0).
\end{align}
For any vector ${\bf a}$
of the form
\begin{align}
{\bf a}  
= 
(a_0, a_{m}, 
a_1, a_{m-1},
\dots,
a_m, a_0), \label{a-symmetric}
\end{align}
it follows that 
\begin{align}
{\bf a} ^T {\bf y}
= &
y_0 a_0 -y_{m} a_m \notag\\
&+
y_1 a_1 -y_{m-1} a_m-1 \notag\\
&+
\dots +
\notag\\
&+y_m a_m -y_0 a_0 \notag\\
=& 0.
\end{align}
From \eqref{k-sym} and $\binom N{g+2bk} =\binom N{2bk}$,
 every odd row of $A$ has the form of \eqref{a-symmetric}.
 Hence, every odd row is orthogonal to ${\bf y}$, 
and it suffices to find a ${\bf y}$ that is orthogonal to every even row in $A$. 

For any vector ${\bf b}$
of the form
\begin{align}
{\bf b}  
= 
(b_0, -b_{m}, 
b_1, -b_{m-1},
\dots,
b_m, -b_0), \label{b-symmetric}
\end{align}
it follows that 
\begin{align}
{\bf b} ^T {\bf y}
= &
y_0 b_0 + y_{m} b_m \notag\\
&+
y_1 b_1 + y_{m-1} b_m-1 \notag\\
&+
\dots +
\notag\\
&+y_m b_m + y_0 b_0,
\end{align}
and the constraint that 
${\bf b} ^T {\bf y}
= 0$ is then equivalent to
showing that 
\begin{align}
   \sum_{k=0}^m y_k b_k = 0.
\end{align}
Hence, we consider a submatrix of $A$ given by 
\begin{align}
  B =
  \sum_{
  \substack{
  0 \le j \le d-1\\
  j \ {\rm odd}
  }  }  
  \sum_{ k = 0}^m
  K^N_{2bk}(j) / \binom N {2bk} |j\>\<k|.
\end{align}
Then, any $(y_0,\dots, y_m)$ in the nullspace of $B$ gives a vector
${\bf y} = (y_0, -y_{m}, 
y_1, -y_{m-1},
\dots,
y_m, -y_0)$ that is also in the nullspace of $A$.
To construct a valid generalized $(b,g)$-PI code, it suffices to require that 
$(y_0,\dots, y_m)$ is a non-negative vector, that is,
$(y_0,\dots, y_m) \ge 0$.

When $d=2t+1$ and $m=t$, the matrix $B$ has $t+1$ columns and $t$ rows. Then, $B$ is guaranteed to have a non-trivial nullspace.
Numerically, in all of the examples we explored, we find that $B$ is a full-rank matrix, and its nullspace of $B$ always admits non-negative vectors. Namely, $B = ({\bf e} , \bar B)$, where $\bar B$ is a square invertible matrix, 
and hence a valid $(y_0,\dots, y_m)$
is proportional to $\bar B^{-1}(y_0,\dots, y_m)$.
Hence, by setting $m=t$,
it is often possible to have generalized bg codes.


We like to prove that this generalized $(b,g)$-PI code has distance at least $2t+1$, thereby correcting $t$ errors, whenever
(1) $g,(b-g) \ge 2t+1$ and (2) $m\ge t$.
Note that condition (1) imposes the bit-flip distance, and hence it remains to show the non-deformation condition for the phase errors. From the above A-matrix type of argument, we just need to ascertain that the vector 
\begin{align}
\gamma_{b,g,m} = 
(\gamma_{b,g,m,0}^2,\dots, 
\gamma_{b,g,m,m}^2) 
\end{align}
resides in the null space of the matrix $B$, that is, $B \gamma_{b,g,m} = 0$.
For this, we need to show that the following is true for all $x=1,\dots, t$:
\begin{align}
    \sum_{k=0}^m
   \frac{ K^N_{2bk}(2x-1) }{ \binom N {2bk} }
    \binom m k
    \gamma_{b,g,m,k}^2 = 0.
\end{align}
The system of linear equations can equivalently be written as
\begin{align}
    \sum_{k=0}^m
   \binom m k
   K^N_{2x-1}(2bk) 
    \gamma_{b,g,m,k}^2 = 0,
\end{align}
which is
\begin{align}
    \sum_{k=0}^m 
   \binom m k
   K^N_{2x-1}(2bk)  
   (N/2b-g/b)_{(m-k)}(N/2b)_{(k)} = 0,
\end{align}
It then remains to prove the following combinatorial lemma.
\begin{lemma}
\label{lem:combi}
Let $b, g, m$ and $x$ be positive integers 
where $x \le m$.
Let $N = 2 b m + g $.
Then
\begin{align}
    S \coloneqq \sum_{k=0}^m 
   \binom m k
   K^N_{2x-1}(2bk)  
   (N/2b-g/b)_{(m-k)}(N/2b)_{(k)} = 0.\notag
\end{align}
\end{lemma}
\begin{proof}
The Krawtchouk polynomial
$K^N_{2x-1}(2bk)$ can be written as a coefficient of a generating function, in the sense that 
\begin{align}
K^N_{2x-1}(2bk)  
=
[t^{2x-1}]
(1-t)^{2bk}
(1+t)^{N-2bk}.
\end{align}
Now let 
$\alpha = N/2b-g/b$
and 
$\beta = N/2b$.
Note that $\alpha+\beta 
= 2N/2b-g/b = (N-g)/b= 2m$.
Also, we can write 
\begin{align}
    \binom m k
   \alpha_{(m-k)}
   \beta_{(k)} 
   &= m! 
   \binom \alpha {m-k}
   \binom \beta {k}.
\end{align}
Then we can write
\begin{align}
S 
&= m!
\sum_{k = 0}^m
  \binom \alpha {m-k}
   \binom \beta {k}
[t^{2x-1}]
(1-t)^{2bk}
(1+t)^{N-2bk}
\notag\\
&= m!
[t^{2x-1}]
(1+t)^N
\sum_{k = 0}^m
  \binom \alpha {m-k}
   \binom \beta {k}
( (1-t) / (1+t) )^{2bk}
.   
\end{align}
Now
\begin{align}
   \sum_{k = 0}^m 
   \binom \alpha {m-k}
   \binom \beta {k}
   y^k
   =
   [s^m]
   (1+s)^\alpha
   (1+sy)^\beta
   .
\end{align}
Identifying $y=(1-t)/(1+t)$,
to prevent higher powers of
$(1-t)/(1+t)$ to appear in the generating function when we make the substitution $y=(1-t)/(1+t)$, we have that 
 whenever $x \le m$, 
\begin{align}
    S&= 
m! [t^{2x-1} s^m ] 
(1+t)^N (1+s)^\alpha (1+s  ((1-t)/(1+t))^\beta
\notag\\
&=
m! [t^{2x-1} s^m ]  E(t,s)
\end{align}
where
\begin{align}
   E(t,s) 
   = 
   (1+s)^\alpha
   ( 
   (1+t)^{2b}
   +
   s (1-t)^{2b} )^\beta
\end{align}
is a formal power series in the variables $t$ and $s$, which allows us to write 
\begin{align}
    E(t,s) = \sum_{u,v \ge 0 } e_{u,v} t^u s^v.
\end{align}  
Now 
\begin{align}
    E(-t,1/s)
    &=
   (1+1/s)^\alpha
   ( 
   (1-t)^{2b}
   +
   s^{-1} (1+t)^{2b} )^\beta
   \notag\\
   &=
   s^{-\alpha - \beta}
   (s+1)^\alpha
   ( 
   s(1-t)^{2b}
   +
   (1+t)^{2b} )^\beta
   \notag\\
   &=
   s^{-\alpha - \beta}
   E(t,s)\notag\\
   &=
   s^{-2m}
   E(t,s). \label{involution}
\end{align}
Using \eqref{involution},
it follows that
\begin{align}
\sum_{u',v' \ge 0 } 
e_{u,v} (-t)^{u'} s^{-v'}
=
\sum_{u,v \ge 0 } 
    e_{u,v} t^u s^{v-2m}.
\end{align}
We are interested in the scenario where 
$v-2m=-m$
and
$-v' = -m$,
which is equivalent to
$v' = v = m$. 
Then we equate the coefficients of the left and right side of the above equation, focusing on the coefficient of $s^{-m} t^u$ for odd $u$.
Then we get
\begin{align}
    e_{u,m} = - e_{u,m}
\end{align}
for odd $u$, which means that 
$e_{u,m} = 0 $ for odd $u.$ Recalling our constraint $x \le m$, which means that we also need to have $u\le 2m-1$, this proves the result.
\end{proof}

Since Lemma \ref{lem:combi} is true, this shows that our $(b,g,m)$-PI code has the stipulated distance property.


        

\section{Fidelity measure for Hadamard gate process}
\label{app:fidelity_measure}
In this section, we discuss how to evaluate the fidelity measure using process matrices. The process matrix gives a convenient way of describing the evolution of quantum states, especially when dealing with mixed states or open quantum systems where the state is described by a density matrix. In the case of a unitary operation $\rho' = U \rho U^{\dagger}$, if we vectorize the density matrix $\rho$ using the row-stacking operation, the action of the superoperator can be expressed in a matrix form as 
\begin{align}
    |\rho'\rangle\rangle=\mathcal{U}|\rho\rangle\rangle,
\end{align}
where $\mathcal{U}=U \otimes U^{\ast}$. $U^{\ast}$ is the complex conjugate of $U$ and $\otimes$ denotes the Kronecker product.

In our case, for the preparation mapping described in Eq.~(\ref{eq:full_prep_lGPG}), the global rotations $R$ are unitary and their associated superoperators are represented in the Kronecker product form. Though the erroneous GPG mapping in Eq.(\ref{eq:prep_linear_err}) is non-unitary, it does not induce coherence between different state components, and can thus be constructed as a diagonal matrix in the superoperator formalism. Consequently, the superoperator associated with the preparation mapping $\mathcal{E}_{\rm l-GPG}$ is given by 
\begin{align}
    E_{\rm l-GPG}=\prod_{p=P}^1 [R(\theta_p,\xi_p,\gamma_p) \otimes R^{\ast}(\theta_p,\xi_p,\gamma_p)]\mathcal{G}_p(\phi_p),
\end{align}
where $\mathcal{G}_p(\phi_p) $ is a $(N+1)^2 \times (N+1)^2$ diagonal matrix that has elements
\begin{align}
    \bra{D_n^N}\bra{D_m^N} \mathcal{G}_p(\phi_p) \ket{D_n^N} \ket{D_m^N} = f_{n,m}(\phi_p)
\end{align}
on each diagonal position accordingly. 

Analogously, the superoperator associated with the reverse mapping of the preparation $\mathcal{E}^{\rm R}_{\rm l-GPG}$ is given by 
\begin{align}
    E^{\rm R}_{\rm l-GPG}=\prod_{p=1}^P \mathcal{G}_p(-\phi_p)[R^{\dagger}(\theta_p,\xi_p,\gamma_p) \otimes (R^{\dagger}(\theta_p,\xi_p,\gamma_p))^{\ast}].
\end{align}

The middle phase gate in Eq.~(\ref{eq:Hdecomp}) is unitary but has an overall fidelity factor $F_{\rm ph}$. We express its superoperator as
\begin{align}
   E_{\rm ph}=F_{\rm ph}C_{N-1}(Z) \otimes C^{\ast}_{N-1}(Z). 
\end{align}

Thus, the actual mapping for the implementation of the logical Hadamard gate in the superoperator formalism is given by
\begin{align}
    E = E_{\rm l-GPG} E_{\rm ph} E^{\rm R}_{\rm l-GPG}.
\end{align}
Since our primary concern is whether the implemented Hadamard gate acts correctly in the logical subspace, we project $E$ from the Dicke subspace onto the logical $\{\ket{0_{\rm L}},\ket{1_{\rm L}}\}$ subspace, resulting in
\begin{align}
    E_{\rm L}=\left( \begin{array}{cccc} E_{00,00} & E_{00,01} & E_{00,10} & E_{00,11} \\
                               E_{01,00} & E_{01,01} & E_{01,10} & E_{01,11} \\
                               E_{10,00} & E_{10,01} & E_{10,10} & E_{10,11} \\
                               E_{11,00} & E_{11,01} & E_{11,10} & E_{11,11} \\
            \end{array}
    \right),
\end{align}
where $E_{x'y',xy}=\bra{x'_{\rm L}}\bra{y'_{\rm L}}E\ket{x_{\rm L}}\ket{y_{\rm L}}$. 
Following the method in  Ref.~\cite{gilchrist2005distance}, the process fidelity between $\mathcal{E}_{\overline{H}}$ and $\overline{\mathcal{H}}$ can be simply calculated as
\begin{align}
    F_{\rm pro}(\mathcal{E}_{\overline{H}},\overline{\mathcal{H}}) = \frac{1}{8} \sum_j {\rm Tr}[\overline{H}_{\rm ide} U_j^{\dagger} \overline{H}_{\rm ide} \mathcal{E}_{\overline{H}}(U_j)],
\label{eq:process_fidelity}
\end{align}
where $\{U_j\}$ are orthonormal bases of unitary operators. Here, in the logical subspace, we select $U_0=I_2, U_1=X, U_2=Y, U_3=Z$ (where $I_2$, $X$, $Y$ and $Z$ are the identity and Pauli matrices). Substitution into Eq.~(\ref{eq:process_fidelity}) gives
\begin{widetext}
\begin{align}
     F_{\rm pro}(\mathcal{E}_{\overline{H}},\overline{\mathcal{H}}) &= \frac{1}{8} ({\rm Tr}[\mathcal{E}_{\overline{H}}(I_2)] +{\rm Tr}[Z\mathcal{E}_{\overline{H}}(X)] +{\rm Tr}[-Y\mathcal{E}_{\overline{H}}(Y)] +{\rm Tr}[X\mathcal{E}_{\overline{H}}(Z)] ) \notag\\
    &=\frac{1}{8} \left[ \left( \begin{array}{cccc} 1&0&0&1 \end{array} \right) E_{\rm L}  \left( \begin{array}{c} 1\\0\\0\\1 \end{array} \right) + \left( \begin{array}{cccc} 1&0&0&-1 \end{array} \right) E_{\rm L}  \left( \begin{array}{c} 0\\1\\1\\0 \end{array} \right)  \right.
      \left. +\left( \begin{array}{cccc} 0&i&-i&0 \end{array} \right) E_{\rm L}  \left( \begin{array}{c} 0\\-i\\i\\0 \end{array} \right) +\left( \begin{array}{cccc} 0&1&1&0 \end{array} \right) E_{\rm L}  \left( \begin{array}{c} 1\\0\\0\\-1 \end{array} \right)   \right]  \notag\\
     &=\frac{1}{8} (E_{00,00}+E_{00,11}+E_{11,00}+E_{11,11}+E_{00,01}+E_{00,10}-E_{11,01}-E_{11,10} \notag\\
     &-E_{01,01}+E_{01,10}+E_{10,01}-E_{10,10}+E_{01,00}-E_{01,11}+E_{10,00}-E_{10,11})
\end{align}
\end{widetext}

\section{Non-linear GPGs for switching between an even-odd code and a non even-odd code}
In this section we describe how to directly implement code switching using logical CNOT gates which works even when the PI code is not an even-odd code, as is the case for the PI-7 code. This requires use of non-linear GPGs and to explain the mechanism we first review how this highly non-linear spin gate is implemented from elementary interactions.   
\label{CNOTNLGPG}
\subsection{Implementing the non-linear GPG}
The dispersive interaction between spins and the bosonic mode we consider takes the form
 $H=g\hat{a}^{\dagger}\hat{a}\otimes \hat{w}_{\Gamma}$.
 Evolving the dispersive interaction $H$ for a time $t$ generates the operator 
 $
 R(\theta \hat{w}_{\Gamma})\coloneqq e^{i\theta  \hat{w}_{\Gamma} \otimes \hat{a}^{\dagger}\hat{a}}
 $
 where $\theta=gt$. This applies a rotation in phase space by an amount proportional to the eigenvalues associated to eigenspaces of the Hermitian operator $\hat{w}$ acting on subsystem $\Gamma$.
 Note it is possible to generate $R(-\theta \hat{w}_{\Gamma})$ by reversing the coupling strength $g\rightarrow -g$ which can be done in some physical setups by e.g.  changing the sign of detuning of the cavity mode from the spin transition frequency. In this section we consider mode-spin interactions $H$ where the spin operator is $\hat{w}_{\Gamma}=\hat{J}^z_{\Gamma}$. This could arise from a native coupling \cite{johnsson2020geometric}, or, as described in Sec.~\ref{stateprepapp}, is locally equivalent to the interaction involving the Hamming weight operator.

 Under conjugation by rotations, displacements can be made conditional on the spin states:
 \begin{equation}
 D(\alpha e^{i\theta \hat{w}_{\Gamma}})
 =
 R(\theta \hat{w}_{\Gamma})D(\alpha)R(-\theta \hat{w}_{\Gamma}).
 \end{equation}
 We can similarly make mode rotations conditional on the spin states using the operator
 \begin{align}
 \Lambda_{\Gamma}(\alpha,\theta)&\coloneqq D(\alpha)R(\theta \hat{w}_{\Gamma})D(-\alpha)R(-\theta \hat{w}_{\Gamma}).
 \end{align}
 A composition of displacement and conditional displacement operators around a closed trajectory in phase space, provides  
 for an identity operator on the mode and a non-linear geometric phase gate on the spins:
 \begin{equation}
 \begin{array}{lll}
 U_{\rm nl-GPG}(\theta,\phi,\chi)&=&D(-\beta)R(\theta \hat{w}_{\Gamma})D(-\alpha)R(-\theta \hat{w}_{\Gamma}) \nonumber \\
 &&\times \,\, D(\beta)R(\theta \hat{w}_{\Gamma})D(\alpha)R(-\theta \hat{w}_{\Gamma}) \nonumber \\
 &=&e^{-i 2 \chi\sin(\theta \hat{w}_{\Gamma}+\phi)} .
 \end{array}
 \end{equation}
 \label{seq}
 $\phi = 
 {\rm arg}(\alpha) - 
 {\rm arg}(\beta)$.
 If the mode begins as the vacuum state then the first rotation operator is not needed.

\subsection{Implementing the $\rm{C}_{\textsf{B}}{\rm X}_{\textsf{A}}$ gate}
\label{CBA}

First, consider the gate $\rm{C}_{\textsf{B}}{\rm X}_{\textsf{A}}$ which has the PI code as control and a stabiliser code as target.
The stabiliser code is assumed to be an even-odd code, while 
The PI code has its 
logical zero codeword $|0_B\>$
and logical one codeword $|1_B\>$
consisting of superpositions of states with Hamming weight $0 {\ \rm mod\ }{q}$
and $s {\ \rm mod\ }{q}$ respectively. In the case of the PI-7 code, $q=5$ and $s=2$.

Here, we construct a controlled geometric phase gate that traverses a path in phase space with zero area if the control has weight $0 {\ \rm mod\ }{q}$, while producing the area required to generate a transversal $X$ gate on the target if the control has $s {\ \rm mod\ }{q}$.

If we pick $\gamma=s\pi/q$, then the action of the controlled displacement operator $\Lambda_B(\alpha,s\pi/q)$ on the codespace of the PI code is
\begin{align} 
&\Lambda_B(\alpha,s\pi/q)(c_0 |0_B\>+c_1|1_B\>)\otimes |\psi_{\rm mode}\>\notag\\
=&c_0 |0_B\>\!\otimes\! |\psi_{\rm mode}\>+c_1|1_B\>\!\otimes \!D(\alpha(1-e^{2is\pi/q}))|\psi_{\rm mode}\> .
\end{align} 
We can use these controlled displacements between the PI code and the mode to produce a GPG on the stabiliser code conditional on the logical state of PI code:
\begin{align} 
\Lambda_B(U_{\rm nl-GPG}(\theta,\phi,\chi)) \coloneqq & \Lambda_B(-\beta,s\pi/q)
R(\theta \hat{J}_{\textsf{A}}^z)\Lambda_B(-\alpha,s\pi/q)\notag\\
&\times R(-\theta \hat{J}_{\textsf{A}}^z)\Lambda_B(\beta,s\pi/q)R(\theta \hat{J}_{\textsf{A}}^z)\notag\\
&\times \Lambda_B(\alpha,s\pi/q)R(-\theta \hat{J}_{\textsf{A}}^z), 
\end{align}
where $\chi=|\alpha \beta| |1-e^{2is\pi/q}|^2
= 4 |\alpha \beta| \sin^2(s \pi/q)$ and $\phi = {\rm arg}(\alpha) - {\rm arg}(\beta)$. 
The action of this conditional nl-GPG on the codespace of the PI code and stabiliser code is
\begin{align}
&\Lambda_B(U_{\rm nl-GPG}(\theta,\phi,\chi))(c_0 |0_B\>+c_1|1_B\>)\otimes |\psi_A\>\notag\\
=&c_0 |0_B\>\otimes |\psi_A\>+c_1|1_B\>\otimes e^{-i2\chi \sin(\theta \hat{J}^z_A + \phi)}|\psi_A\>.
\end{align}
We pick $\alpha = \beta$ so that $\phi = 0$.
 With the choice $\theta=\pi$, when the control is in $|1_B\>$, the mode experiences a trajectory in phase space that is a rotated square with area $\chi$, and when the control is in $|0_B\>$ the trajectory has zero area. If we further pick $\chi=\pi/4$, i.e. $\alpha=\sqrt \pi/ ( 4 | \sin(s\pi/q)| ) $, then $e^{-i2\chi \sin(\theta \hat{J}^z_A )}=i\prod_{j=1}^7 Z_j=i\bar{Z}_A$, where $\bar{Z}_A$ is the logical $Z$ operator on the stabiliser code. 

The final gate is then
\begin{equation}
\begin{array}{lll}
\rm{C}_{\textsf{B}}{\rm X}_{\textsf{A}}&=&
\bar{S}_A\bar{H}_A
\Lambda_1(U_{\rm nl-GPG}(\pi,0,\frac{\pi}{4}))
 \bar{H}_A,
\end{array}
\end{equation}
where $\bar{H}_A=H^{\otimes |A|}$ is the transversal logical Hadamard and $\bar{S}_A=\bar{Z}S^{\otimes |A|}$ is the logical $S=Z(\pi/2)$ gate on the stabiliser code.

\subsection{Implementing the $\rm{C}_{\textsf{A}}{\rm X}_{\textsf{B}}$ gate}
\label{CAB}
Second, consider the gate $\rm{C}_{\textsf{A}}{\rm X}_{\textsf{B}}$ which has the stabiliser code as the control and the PI code as target.
We can use the same procedure as in Sec.~\ref{CBA}, with the roles of $A$ and $B$ reversed, except we choose $g\tau=\gamma=\pi$.
In this case the action angle in phase space will be $\chi=4|\alpha|^2$ and the geometric phase gate will have the parameters $\phi=\pi$, $\theta=\pi$, and $\alpha=\beta=\sqrt{\pi}/4$. This will achieve the operation $-i\bar{Y}_B$ if the control is in $\ket{1_A}$, and acts trivially if the control is in $\ket{0_A}$. In summary,
\begin{equation}
\begin{array}{lll}
\rm{C}_{\textsf{A}}{\rm X}_{\textsf{B}}&=&
\Lambda_A(U_{\rm nl-GPG}(\pi,\pi,\frac{\pi}{4}))\\
&=&\Lambda_A(-\frac{\sqrt{\pi}}{4},\pi)e^{i\frac{\pi}{2}\hat{J}^x_B}R(\pi \hat{J}_{\textsf{B}}^z)e^{-i\frac{\pi}{2}\hat{J}^x_B}\Lambda_A(-\frac{\sqrt{\pi}}{4},\pi)\\
&&\times e^{i\frac{\pi}{2}\hat{J}^x_B}R(-\pi \hat{J}_{\textsf{B}}^z)e^{-i\frac{\pi}{2}\hat{J}^x_B}\Lambda_A(\frac{\sqrt{\pi}}{4},\pi)e^{i\frac{\pi}{2}\hat{J}^x_B}R(\pi \hat{J}_{\textsf{B}}^z)\\
&&\times e^{-i\frac{\pi}{2}\hat{J}^x_B}\Lambda_A(\frac{\sqrt{\pi}}{4},\pi)e^{i\frac{\pi}{2}\hat{J}^x_B}R(-\pi \hat{J}_{\textsf{B}}^z)e^{-i\frac{\pi}{2}\hat{J}^x_B}.
\end{array}
\label{eq:CAB}
\end{equation}

\end{document}